%% file: paper.tex
\documentclass[letterpaper, 10 pt, conference]{ieeeconf}

\IEEEoverridecommandlockouts
\usepackage{cite}

\usepackage{amsthm,amsmath,amssymb,amsfonts}
\usepackage{bm}
\usepackage{algorithmic}
\usepackage[ruled, noline, linesnumbered]{algorithm2e}
\usepackage{graphicx}
\usepackage{xcolor}

\usepackage{array}
\usepackage{multirow}
\usepackage{diagbox}

\makeatletter
\let\MYcaption\@makecaption
\makeatother

\usepackage[font=footnotesize]{subcaption}

\makeatletter
\let\@makecaption\MYcaption
\makeatother

\usepackage{textcomp}
\def\BibTeX{{\rm B\kern-.05em{\sc i\kern-.025em b}\kern-.08em
T\kern-.1667em\lower.7ex\hbox{E}\kern-.125emX}}

\usepackage{lipsum}

\newcommand\copyrighttext{
\footnotesize \copyright\ 2020 IEEE. Personal use of this material is permitted. Permission from IEEE must be obtained for all other uses, in any current or future media, including reprinting/republishing this material for advertising or promotional purposes, creating new collective works, for resale or redistribution to servers or lists, or reuse of any copyrighted component of this work in other works. DOI: 10.23919/ACC45564.2020.9147953
}
\newcommand\copyrightnotice{
\begin{tikzpicture}[remember picture,overlay]
\node[anchor=south,yshift=0.2in] at (current page.south) {\parbox{\dimexpr\textwidth-\fboxsep-\fboxrule\relax}{\copyrighttext}};
\end{tikzpicture}\vspace{-0.825\baselineskip}
}


\theoremstyle{plain}
\newtheorem{theorem}{Theorem}

\theoremstyle{definition}

\def\({\left(}
\def\){\right)}
\def\[{\left[}
\def\]{\right]}

\def\dbf{{\bf d}}

\def\Kbf{{\bf K}}  
\def\ubf{{\bf u}}
\def\xbf{{\bf x}}
\def\deltabf{\bm{\delta}}
\def\Phibf{\bm{\Phi}}
\def\Rmbb{\mathbb{R}}  
\def\Hcal{\mathcal{H}}  
\def\Lcal{\mathcal{L}}  
\def\Rcal{\mathcal{R}}  
\def\Scal{\mathcal{S}}  
\newcommand{\blue}[1]{{#1}}

\newcommand{\SFpair}{\{\Phibf_\xbf, \Phibf_\ubf\}}

\def\mat#1{\begin{bmatrix}#1\end{bmatrix}}
\def\t{[t]}
\def\tn{[t+1]}
\def\tm{[t-1]}
\def\fig#1{Fig.~\ref{fig:#1}}
\def\subfig#1#2{Fig.~\ref{fig:#1}(\subref{subfig:#1-#2})}
\def\sec#1{Section~\ref{sec:#1}}
\def\tab#1{Table~\ref{tab:#1}}
\def\thm#1{Theorem~\ref{thm:#1}}
\def\eqn#1{\eqref{eqn:#1}}

\def\st{{\rm s.t.}}
\def\OptConsSep{&&\quad}

\newcommand{\OptMin}[2]{
\begin{alignat}{2}
\text{minimize}\ &\ #1 \nonumber \\
\st\ #2
\end{alignat}
}

\newcommand\OptCons[3]{
&\ #1
\ifx\\#2\\ \else \OptConsSep #2 \fi%
\ifx\\#3\\ \nonumber \else \label{eqn:#3} \fi%
}

\title{\LARGE \bf Deployment Architectures for Cyber-Physical Control Systems}

\author{Shih-Hao Tseng and James Anderson
\thanks{Shih-Hao Tseng and James Anderson are with the Division of Engineering and Applied Science, California Institute of Technology, Pasadena, CA 91125, USA.  Emails: {\tt\small \{shtseng,james\}@caltech.edu}}
\copyrightnotice
}

\usepackage{tikz}

\begin{document}

\maketitle
\thispagestyle{empty}
\pagestyle{empty}

\bstctlcite{IEEE_BSTcontrol}

\input{abstract}

\input{notation}
\input{introduction}

\input{background}

\input{architecture}
\input{centralized}
\input{global-state}
\input{distributed}
\input{comparison}
\input{future-directions}

\bibliographystyle{IEEEtran}
\bibliography{Test}

\end{document}

%% file: abstract.tex
\begin{abstract}
We consider the problem of how to deploy a controller to a (networked) cyber-physical system (CPS). Controlling a CPS is an involved task, and synthesizing a controller to respect sensing, actuation, and communication constraints is only part of the challenge. In addition to controller synthesis, one should also consider how the controller will work in the CPS. Put  another way, the cyber layer and its interaction with the physical layer need to be taken into account.

In this work, we aim to bridge the gap between theoretical controller synthesis and practical CPS deployment. We adopt the system level synthesis (SLS) framework to synthesize a state-feedback controller and provide a deployment architecture for the standard SLS controller. Furthermore, we derive a new controller realization for open-loop stable systems and introduce four different architectures for deployment, ranging from fully centralized to fully distributed. Finally, we compare the trade-offs among them in terms of robustness, memory, computation, and communication overhead.
\end{abstract}

%% file: notation.tex
\paragraph*{Notation and Terminology}

Let $\Rcal\Hcal_{\infty}$ denote the set of stable rational proper transfer matrices, and $z^{-1}\Rcal\Hcal_{\infty} \subset \Rcal\Hcal_{\infty}$ be the subset of strictly proper stable transfer matrices. Lower- and upper-case letters (such as $x$ and $A$) denote vectors and matrices respectively, while bold lower- and upper-case characters and symbols (such as $\ubf$ and ${\Phibf_\ubf}$) are reserved for signals and transfer matrices. Let $A^{ij}$ be the entry of $A$ at the $i^{\rm th}$ row and $j^{\rm th}$ column. We define $A^{i\star}$ as the $i^{\rm th}$ row and $A^{\star j}$ the $j^{\rm th}$ column of $A$. We use ${\Phi_u}[\tau]$ to denote the $\tau^{\rm th}$ spectral element of a transfer function ${\Phibf_\ubf}$, i.e., ${\Phibf_\ubf} = \sum\limits_{\tau=0}^{\infty} z^{-\tau} {\Phi_u}[\tau]$.

We briefly summarize below the major terminology:
\begin{itemize}
\item {\bf Controller model:} a (linear) map from the state vector to the control action.
\item {\bf Realization:} a control block diagram/state space dynamics based on some controller model.
\item {\bf Architecture:} a cyber-physical system structure built from basic components.
\item {\bf Synthesize:} derive a controller model (and some realization).
\item {\bf Deploy/Implement:} map a controller model (through a realization) to an architecture.
\end{itemize}

%% file: introduction.tex
\section{Introduction}\label{sec:introduction}
We consider a linear time-invariant (LTI) system with a set of sensors $s_i, i = 1, \dots, N_x$ and a set of actuators $a_k, k = 1, \dots, N_u$, with plant dynamics
\begin{align}
x\tn = A x\t + B u\t + d_x\t
\label{eqn:system-dynamics}
\end{align}
where $x[t] \in \Rmbb^{N_x}$ is the state vector, $u[t] \in \Rmbb^{N_u}$ is the control, and $d_x[t] \in \Rmbb^{N_x}$ is the disturbance. Suppose the system is open-loop stable, i.e., $(zI-A)^{-1}B \in \Rcal\Hcal_{\infty}$. The goal of this paper is to address how a state-feedback controller can be deployed to this system and what the corresponding  cyber-physical structures and trade-offs are.

\begin{figure}
\centering
\includegraphics{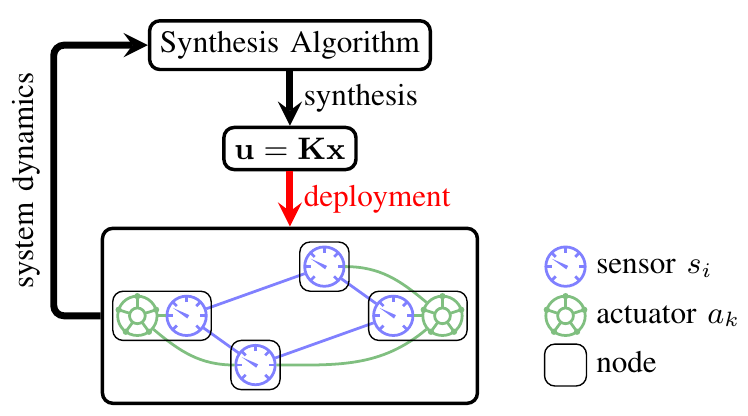}
\caption{A model-based system control scheme consists of two phases -- synthesis and deployment. We can adopt SLS in the synthesis phase to obtain the optimal controller model $\Kbf$, and this paper focuses on how to deploy $\Kbf$ to the underlying CPS.}
\label{fig:flow-chart}
\end{figure}

A model-based approach to control design involves two phases: the \emph{synthesis phase} and the \emph{deployment phase}, as illustrated in \fig{flow-chart}. In the synthesis phase, the control engineer derives the desired controller model by some synthesis algorithm based on a  model of the system dynamics, a suitable objective function, and operating constraints on sensing, actuation, and communication capabilities. The \emph{optimal} controller in the model-based sense is thus the model achieving the best objective value.

Synthesizing an optimal controller for a cyber-physical system is, in general, a daunting task. Recently, a framework named \emph{system level synthesis (SLS)} was proposed  to facilitate distributed controller synthesis for large-scale (networked) systems \cite{doyle2017system,wang2019system,AndDLM19}. Instead of designing the controller itself, SLS directly synthesizes desired closed-loop system responses subject to system level constraints, such as localization constraints \cite{WanMD18} and state and input constraints \cite{CheA19}. Using the closed-loop system response, SLS derives the optimal linear controller model, which admits multiple (mathematically equivalent)  control block diagrams/state space realizations (or simply  \emph{realizations}) \cite{wang2019system,anderson2017structured}.

On the other hand, the deployment phase (often referred to as implementation) is concerned with how to map the derived controller model/realization to the target system.
We can usually implement one controller realization by multiple different \emph{deployment architectures} (or simply  \emph{architectures}). Although all  architectures lead to the ``optimal'' controller, they can differ in  aspects other than the objective, for example; memory requirements, robustness to failure, scalability, and financial cost. Therefore, it is important to consider the \emph{trade-offs} among those architectures, in order to deploy the most suitable architecture.

Approaches to deployment   vary greatly in the literature. In the control  literature, most work gives the controller design at the realization level and implicitly relies on some interface provided by the underlying CPS for deployment \cite{fink2011robust,schwager2011eyes}. Providing the appropriate abstraction and programming interface is itself a design challenge \cite{lee2008cyber,khaitan2014design,lee2015cyber,hu2016robust}. Some papers also examine the deployment down to the circuit level \cite{shao2006new,Jer14}.
The networking/system community, on the other hand, mostly adopts a bottom-up instead of a top-down approach to system control. It usually involves some carefully designed gadgets/protocols and a coordination algorithm \cite{hill2000system,hamed2018chorus,dhekne2019trackio}.

In this work we take an alternative approach to deployment, which lies between the realization and the circuit level. Rather than binding the design to some specific hardware, we specify the basic components of the system and use them to implement the derived controller realization. As such, we can easily map our architectures to real CPS -- as long as the system supports all basic functions. It is the flexibility of the SLS approach which makes such a process possible.

The paper is organized as follows. In \sec{background}, we briefly review \emph{system level synthesis}, propose a new, simpler control realization, and show it is internally stabilizing. With this  realization, we propose four different partitions and their corresponding deployment architectures in \sec{architecture}, namely, the centralized (\sec{centralized}), global state (\sec{global-state}), naive and memory conservative distributed (\sec{distributed}) architectures. Then, in \sec{comparison} we discuss the trade-offs made by the architectures on robustness, memory, computation, and communication. Finally, we conclude the paper with possible future research directions in \sec{future}.

%% file: background.tex
\section{Synthesis Phase}\label{sec:background}
We briefly review the SLS method for the distributed control synthesis phase. Then we describe two ``standard'' SLS controller realizations and draw attention to their respective architectures for the deployment phase. Finally, we derive a new controller realization for open-loop stable systems which we will show in later sections has many favorable properties in the deployment stage.

\subsection{System Level Synthesis (SLS) -- An Overview}

To synthesize the state-feedback closed-loop controller for the system described in \eqn{system-dynamics}, SLS introduces the \emph{system response} $\SFpair$ transfer matrices. The system response is the closed-loop mapping from disturbance $\dbf_\xbf$ to state $\xbf $ and control action $\ubf$, under the feedback policy $\ubf = \Kbf \xbf$. Compactly this is written as
\begin{align*}
\mat{\xbf \\ \ubf} =
\mat{{\Phibf_\xbf} \\ {\Phibf_\ubf} }
\dbf_\xbf.
\end{align*}
where $\{{\Phibf_\xbf},{\Phibf_\ubf}\}$ have the realizations $\Phibf_\xbf = (zI-A-B\Kbf)^{-1}$ and $\Phibf_\ubf = \Kbf(zI-A-B\Kbf)^{-1}$. It was shown in~\cite{wang2019system} that the set of all achievable internally stabilizing controllers is parametrized by an affine subspace.\footnote{Recent work has shown the equivalence between the SLS closed-loop parameterization and the well-known Youla parameterization \cite{ZheFPLK19}.} Using this property, \emph{system level synthesis} problem takes the form:
\OptMin{
g({\Phibf_\xbf},{\Phibf_\ubf})
}{
\OptCons{
\mat{zI-A & -B}
\mat{{\Phibf_\xbf}\\
{\Phibf_\ubf}
}
=
I
}{}{constraint}\\
\OptCons{
{\Phibf_\xbf},{\Phibf_\ubf} \in z^{-1}\Rcal\Hcal_{\infty}
}{}{}\\
\OptCons{
\mat{{\Phibf_\xbf}\\
{\Phibf_\ubf}} \in \Scal .
}{}{}
}
As long as $g$ is a convex functional, \blue{e.g., $\Lcal_1$, $\Hcal_{\infty}$, or $\Hcal_2$ norms \cite{AndDLM19},}  and $\Scal$ defines a convex set, the SLS problem~\eqn{constraint} is convex. For details on how classical optimal control problems can be formulated in this manner, as well as many other aspects of SLS, the reader is referred to~\cite[\S 2.2]{AndDLM19}. The constraint set $\Scal$ is used to provide spatial and temporal locality to the closed-loop response. This is typically done by enforcing sparsity on the spectral components of $\{{\Phibf_\xbf},{\Phibf_\ubf}\}$ and making them finite-impulse-response (FIR) filters. Full details can be found in~\cite{AndDLM19,WanMD18}.

\begin{figure}
\centering
\includegraphics{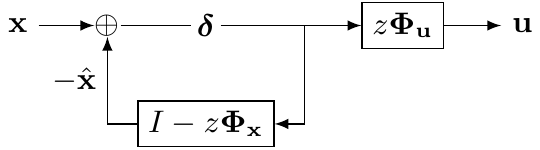}
\caption{SLS control block diagram derived in \cite{doyle2017system}. The signals of interest are $\deltabf  = \xbf -\bf{\hat{x}}$, $\ubf = z \Phibf_\ubf$, and $-\hat{\xbf} = I-z\Phibf_\xbf$. These follow from the block diagram or by taking the z-transform of~\eqn{SLSdefault}.
}
\label{fig:SLS-controller-diagram}
\end{figure}

With~\eqn{constraint} solved, a stabilizing controller is readily obtained from the system response:
\begin{align}\label{eqn:K}
\Kbf = {\Phibf_\ubf} {\Phibf_\xbf}^{-1}.
\end{align}
The controller above is useful for theoretical purposes. However, inverting $\Phibf_\xbf$ is undesirable in most cases as it is heavily dependent on conditioning, and all the structure that $\Phibf_\xbf$ has will likely be lost. Indeed, a controller realization based on the block diagram in \fig{SLS-controller-diagram} is more suitable for distributed control. In state space, it takes the form:
\begin{align}\label{eqn:SLSdefault}
\delta\t &= x\t - \hat{x}\t,\nonumber \\
u\t & = \sum_{\tau \geq 1}\Phi_u[\tau]\delta[t + 1 - \tau],\\
\hat{x}\tn & = \sum_{\tau \geq 2}\Phi_x[\tau]\delta[t + 2 - \tau], \nonumber
\end{align}
where $-\hat{\xbf}$ is the output of the $I-z\Phibf_\xbf$ block. This controller is a disturbance-rejection controller as the signal $\deltabf$ estimates the disturbance signal $\dbf_\xbf$. From a realization perspective, this is better than realizing \eqn{K} by a single block. Unlike \eqn{K}, which inverts $\Phibf_\xbf$, any structure imposed on $\SFpair$ is inherited by the two blocks in \fig{SLS-controller-diagram}, i.e., closed-loop constraints are passed on to the controller. One of the main contributions of this paper is to show that there is another controller realization, which, from a deployment perspective, may be more economical to implement whilst maintaining the desirable properties of~\eqn{SLSdefault}.

\subsection{A Simpler SLS Controller Realization}
We have seen in the previous section how the distributed controller obtained from ${\Phibf_\xbf}$ and ${\Phibf_\ubf}$ can be realized. To give a bit more insight into this design, observe that
\begin{align*}
\Kbf = {\Phibf_\ubf}{\Phibf_\xbf}^{-1} = (z{\Phibf_\ubf})(z{\Phibf_\xbf})^{-1}.
\end{align*}
The realization in \fig{SLS-controller-diagram} then follows by putting $z{\Phibf_\ubf}$ in the forward path and realizing $(z{\Phibf_\xbf})^{-1}$ as the feedback path through the $I-z{\Phibf_\xbf}$ block.

Although the design in \fig{SLS-controller-diagram} is straightforward, and certainly preferable to direct implementation via~\eqn{K}, we note that two convolutions are required to construct $u\t$. Furthermore, the memory required to hold the convolution kernels grows rapidly with horizon length (of the FIR constraint in $\mathcal S$) and number of system states~\cite{anderson2017structured}. We will show that when the plant is open-loop stable, we can simplify the design by replacing a convolution with two matrix-multiplies as in \fig{controller-diagram}.
In \sec{centralized} we carry out an explicit cost comparison between the architectures resulting from \fig{SLS-controller-diagram} and \fig{controller-diagram}.

\begin{theorem}\label{thm:controller}
Let $A \in  \Rmbb^{N_x \times N_x}$ in~\eqn{system-dynamics} be Schur stable. The dynamic state-feedback controller $\ubf = \Kbf \xbf$ realized via
\begin{subequations}\label{eq:controller_time}
\begin{align}
\delta\t =&\ x\t - A x\tm - B u\tm,
\label{eqn:state-space-delta}\\
u\t =&\ \sum\limits_{\tau\geq 1} {\Phi_u}[\tau]\delta[t+1-\tau],
\label{eqn:state-space-u}
\end{align}
\end{subequations}
is internally stabilizing and is described by the block diagram in \fig{controller-diagram}.
\end{theorem}
\begin{proof}
To derive block diagram in \fig{controller-diagram}, observe that any feasible $\SFpair$ pair satisfies \eqn{constraint}, which implies
\begin{align}
{\Phibf_\xbf} = (zI-A)^{-1}(I + B{\Phibf_\ubf}).
\label{eqn:R-to-M}
\end{align}
In other words, the information of ${\Phibf_\xbf}$ is already encoded in ${\Phibf_\ubf}$. Substituting \eqn{R-to-M} into the transfer function $\Kbf={\Phibf_\ubf}{\Phibf_\xbf}^{-1}$, we have
\begin{align}
\Kbf =&\ {\Phibf_\ubf} (I + B{\Phibf_\ubf})^{-1} (zI-A) \nonumber \\
=&\ (z{\Phibf_\ubf}) (I + z^{-1} B (z{\Phibf_\ubf}))^{-1} (I-z^{-1}A).
\label{eqn:new-transfer-function}
\end{align}
The block diagram follows immediately from~\eqn{new-transfer-function}. The time domain dynamics~\eqref{eq:controller_time} are obtained from an inverse z-transform.

To prove internal stability, additional perturbations are added to the closed-loop, and it is shown that they decay with time. Consider the plant and controller connected in feedback as in \fig{internal-stability}. It is sufficient to examine how the internal states $\xbf$, $\ubf$, and $\deltabf$ are affected by those external signals. The signals are related via the following equations:
\begin{align*}
z\xbf =&\ A \xbf + B\ubf + \dbf_\xbf, \\
\ubf =&\ (z{\Phibf_\ubf})\deltabf + \dbf_\ubf, \\
\deltabf =&\ -z^{-1}B (z{\Phibf_\ubf})\deltabf + (I-z^{-1}A)\xbf + \dbf_{\deltabf},
\end{align*}
which, after rearranging, gives
\begin{align*}
\mat{
\xbf \\
\ubf \\
\deltabf
} =
\mat{
{\Phibf_\xbf} & {\Phibf_\xbf} B & (zI-A)^{-1}B(z{\Phibf_\ubf}) \\
{\Phibf_\ubf} & I + {\Phibf_\ubf} B & z{\Phibf_\ubf} \\
z^{-1} & z^{-1}B & I
}
\mat{
\dbf_\xbf \\
\dbf_\ubf \\
\dbf_{\deltabf}
},
\end{align*}
where ${\Phibf_\xbf} = (zI-A)^{-1} (I+B{\Phibf_\ubf})$. If all nine entries in the transfer matrix above are stable, then the (bounded) injected signals $\dbf_\xbf$, $\dbf_\ubf$, and $\dbf_{\deltabf}$ produce  bounded signal $\xbf$, $\ubf$, and $\deltabf$, which implies internal stability. By assumption, $(zI-A)^{-1}B$ is stable, and~\eqn{constraint} constrains ${\Phibf_\xbf}, {\Phibf_\ubf}$ to $z^{-1}\Rcal\Hcal_{\infty}$. It is then easy to verify stability of all nine entries, thus proving the system is internally stable.
\end{proof}

\begin{figure}
\centering
\includegraphics{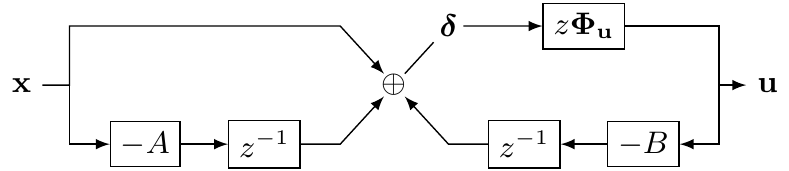}
\caption{A simpler control block diagram corresponding to~\eqref{eq:controller_time}. Note the single convolution $z{\Phibf_\ubf}$ according to $\Kbf = {\Phibf_\ubf}{\Phibf_\xbf}^{-1} = (z{\Phibf_\ubf}) (I + z^{-1} B (z{\Phibf_\ubf}))^{-1} (I-z^{-1}A)$.}
\label{fig:controller-diagram}
\end{figure}

The advantage of \fig{controller-diagram} is that it uses only one convolution $z{\Phibf_\ubf}$. Also, when $\Phibf_\ubf$ has finite impulse response (FIR) with horizon $T$, we have
\begin{align*}
u\t =&\ \sum\limits_{\tau=1}^{T} {\Phi_u}[\tau]\delta[t+1-\tau]
\end{align*}
regardless of whether ${\Phibf_\xbf}$ is FIR or not.

We remark that it is possible to internally stabilize some open-loop unstable systems by the controller in \thm{controller} (\fig{controller-diagram}). Consider a decomposition of the system matrix such that $A = A_u + A_s$ where $A_u$ is unstable and $A_s$ is Schur stable. Then using the robustness results from~\cite{MatWA17}, the controller designed for $(A_s, B)$ will stabilize $(A,B)$ if  $\|A_u \Phibf_\xbf\|<1$ for any induced norm.

\begin{figure}
\centering
\includegraphics{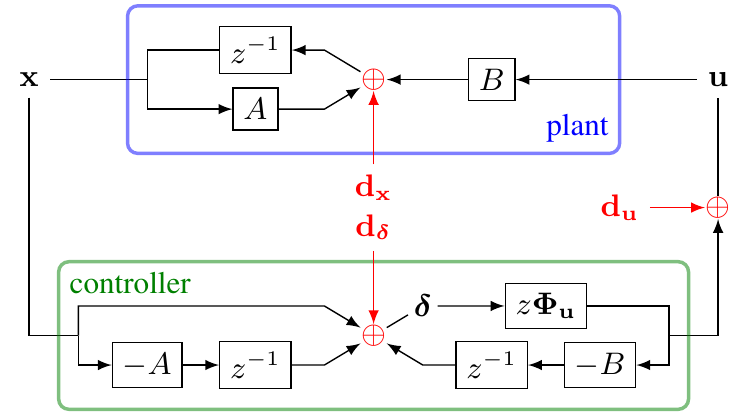}
\caption{To show internal stability, we interconnect the proposed controller (\fig{controller-diagram}) to the plant and examine the effect of injected external signals $\dbf_\xbf$, $\dbf_\ubf$, and $\dbf_{\deltabf}$ on the internal states $\xbf$, $\ubf$, and $\deltabf$.}
\label{fig:internal-stability}
\end{figure}

%% file: architecture.tex
\section{Deployment Architectures}\label{sec:architecture}
We now explore the controller architectures for the deployment phase. The crux of our designs centers on the partitions of the controller realization described in \thm{controller} and illustrated in \fig{controller-diagram}. We first introduce the basic components. Using those components, we propose the centralized, global state, and distributed architectures accordingly. As long as the real system is capable of providing the basic components (or some equivalent parts), mapping the architectures into the system is straightforward.

\subsection{Basic Components}
\begin{figure}
\centering
\subcaptionbox{Computation and storage components\label{subfig:architecture-components-computation-and-storage}}{
\includegraphics{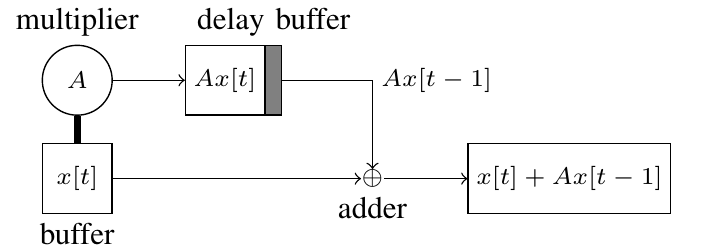}
}\\[0.5\baselineskip]
\subcaptionbox{Communication components\label{subfig:architecture-components-communication}}{
\includegraphics{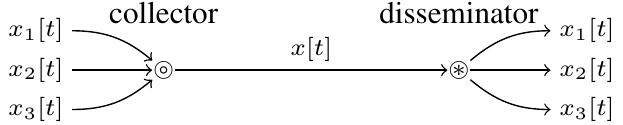}
}
\caption{The basic functions that a node in the system can perform. These functions serve as the basic components of the proposed deployment architectures. They can be categorized into (\subref{subfig:architecture-components-computation-and-storage}) computation and storage components and (\subref{subfig:architecture-components-communication}) communication components.}
\label{fig:architecture-components}
\end{figure}

We assume that the nodes in the target system are capable of performing the functions shown in \fig{architecture-components}. These functions cover variable storage, arithmetic computation, and communication. Below we describe these functions in details.

The essential component is a \emph{buffer}. A buffer is where the system keeps the value of a variable. It can be a memory device such as a register or merely a collection of some buses (wires). On the other hand, a \emph{delay buffer} is the physical implementation of $z^{-1}$ in a block diagram, which keeps the variable received in the current time step and releases its value in the next time step.

For computation, a \emph{multiplier} with a matrix $A$ senses the variable in a targeted buffer
and multiplies it by $A$ as the output. An \emph{adder} performs entry-wise addition of two vectors or matrices of compatible dimensions. Although an adder is a two-input component, we can merge cascaded adders into a multiple-input adder in practice.

A node can also communicate with other nodes through \emph{disseminator-collector pairs}. A \emph{disseminator} sends some parts of a variable to designated nodes. At the receiving side, a \emph{collector} assembles the received parts appropriately to reconstruct the desired variable.

%% file: centralized.tex
\subsection{Centralized Architecture}\label{sec:centralized}
The most straightforward deployment architecture is the centralized architecture, which partitions the block diagram as in \fig{centralized}. It packs all the control functions into one node, the centralized controller. The centralized controller maintains communications with all sensors and actuators to collect state information and dispatch the control signal.

\begin{figure}
\centering
\includegraphics{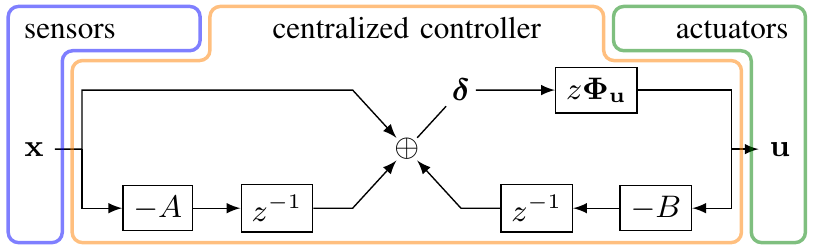}
\caption{The partitions of the block diagram for the centralized architecture. \blue{The architecture of the centralized controller is depicted in \fig{centralized-architecture}.}}
\label{fig:centralized}
\vspace*{\baselineskip}
\includegraphics{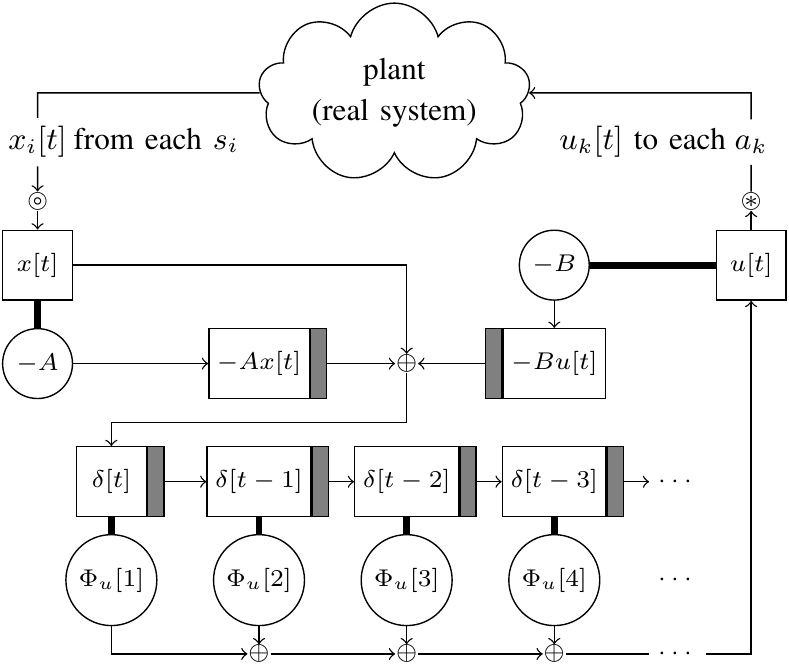}
\caption{The architecture of the centralized controller.}
\label{fig:centralized-architecture}
\end{figure}

\fig{centralized-architecture} shows the architecture of the centralized controller. For each time step $t$, the controller first collects the state information $x_i\t$ from each sensor $s_i$ for all $i = 1,\dots,N_x$. Along with the stored control signal $u$, the centralized controller computes $\delta\t$ as in \eqn{state-space-delta}. $\delta\t$ is then fed into an array of delay buffers and multipliers to perform the convolution \eqn{state-space-u} and generate the control signals. The control signals $u_k\t$ are then sent to each actuator $a_k$ for all $k = 1,\dots, N_u$.

Deploying a synthesized solution to the centralized architecture is simple: We take the spectral components ${\Phi_u}[\tau]$ of ${\Phibf_\ubf}$ from $\tau = 1,2,\dots$ and insert them into the array of the multipliers. Also, $A$ and $B$ are adopted directly from the system model.

We can compare the centralized architecture of \fig{controller-diagram} with the architecture of the original block diagram \fig{SLS-controller-diagram}, which is shown in \fig{original-architecture}. As mentioned in \sec{background}, the original architecture (\fig{original-architecture}) is expensive both computationally and storage-wise.  Specifically, when both ${\Phibf_\xbf}$ and ${\Phibf_\ubf}$ are FIR with horizon $T$, the original architecture depicted in \fig{original-architecture} performs
\begin{equation}
\underbrace{\strut (T-1)N_x^2}_{\Phi_x[\cdot]\delta[\cdot]} + \underbrace{\strut T N_x N_u}_{\Phi_u[\cdot]\delta[\cdot]}\nonumber
\end{equation}
scalar multiplications per time step and needs
\begin{equation}
\underbrace{\strut (T-1)N_x^2}_{\Phi_x[\cdot]} +
\underbrace{\strut T N_x N_u}_{\Phi_u[\cdot]} +
\underbrace{\strut (T + 2) N_x + N_u}_{\delta[\cdot],\ x\t,\ -\hat{x}[t+1], \text{~and~} u\t}\nonumber
\end{equation}
scalar memory locations to store all variables and multipliers. On the other hand, \fig{centralized-architecture} performs
\begin{equation}
\underbrace{\strut N_x^2 + N_x N_u}_{-Ax\t\text{~and~}-Bu\t} +
\underbrace{\strut T N_x N_u}_{\Phi_u[\cdot]\delta[\cdot]}
\label{eqn:centralized-computation-requirement}
\end{equation}
scalar multiplications and needs
\begin{align}
\underbrace{\strut N_x^2 + N_x N_u}_{A\text{~and~}B} +
&\underbrace{\strut 2 N_x }_{-Ax\t\text{~and~}-Bu\t} +\nonumber\\
&\underbrace{\strut TN_x N_u}_{\Phi_u[\cdot]\delta[\cdot]} +
\underbrace{\strut (T + 1) N_x + N_u}_{\delta[\cdot],\ x\t, \text{~and~} u\t}
\label{eqn:centralized-storage-requirement}
\end{align}
storage space, which is more economic when $N_x \geq N_u$, $N_x \geq 2$, and $T > 3$, as is the case for under-actuated systems, and typical for distributed control problems.

We note that the above analysis is based on the most general case -- assuming that all matrices are dense. In reality, $A$ and the spectral components $\Phi_x[\cdot], \Phi_u[\cdot]$ can be sparse and structured, which admits some specialized multiplier and buffer designs to significantly reduce the overhead  (such as \cite{vuduc2005oski,williams2007optimization,bell2008efficient}). This is one reason why the original SLS controller design (\fig{SLS-controller-diagram} and \fig{original-architecture}) is in some cases quite efficient: When $\Scal$ enforces spatial-localization, the spectral elements $\Phi_x[\cdot]$ and $\Phi_u[\cdot]$ exhibit sparsity patterns, and thus computing all $\Phi_x[\cdot]\delta[\cdot]$ could potentially be easier than obtaining $-Ax\t$ and $-Bu\t$.
In most cases, the new architecture (\fig{centralized-architecture}) saves computation by replacing $T - 1$ matrix-vector multiplies with two (sparse) matrix-vector multiplies $-Ax\t$ and $-Bu\t$.

\begin{figure}
\centering
\includegraphics{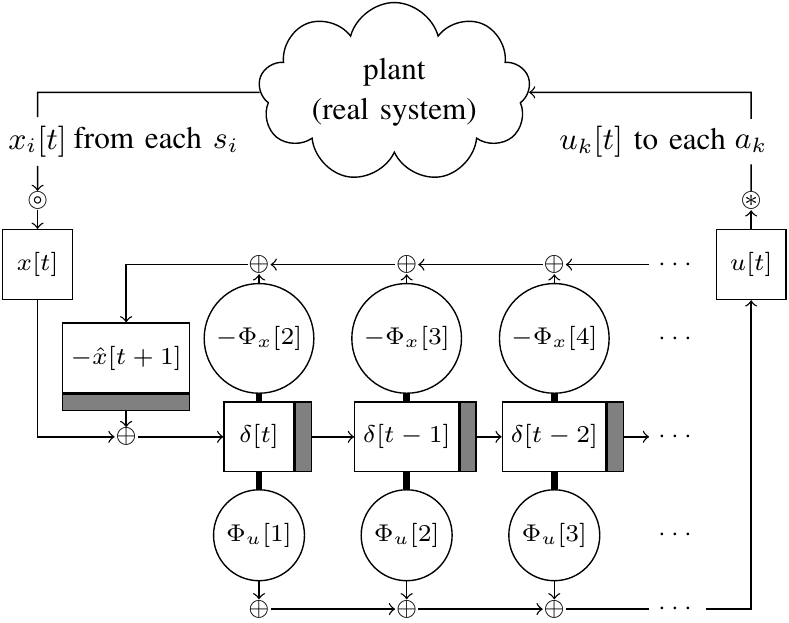}
\caption{The architecture of the original controller model in \fig{SLS-controller-diagram}. We remark that \eqn{constraint} implies ${\Phi_x}[1]=I$ and hence we don't need a multiplier at $\delta[t]$ for the convolution $I-z{\Phibf_\xbf}$. When ${\Phibf_\xbf}$ and ${\Phibf_\ubf}$ are both FIR with horizon $T > 3$, $N_x \geq N_u$, and $N_x \geq 2$, this architecture requires more computation and storage resources than \fig{centralized-architecture}.}
\label{fig:original-architecture}
\end{figure}

Despite the intuitive design, the centralized architecture raises several operational concerns. First, the centralized controller becomes the single point of failure. Also, the scalability of the centralized scheme is poor: The centralized controller has to ensure communication with all the sensors/actuators and deal with the burden of high computational load. In the next subsection, we explore some other system architectures to address these issues.

%% file: global-state.tex
\subsection{Global State Architecture}\label{sec:global-state}
To avoid overloading the centralized controller, we can offload the computations to the nodes in the system, which results in the global state architecture. \fig{global-state} shows the partitions of the block diagram. We have a centralized global state keeper (GSK) which keeps track of the global state $\delta\t$ instead of the raw state $x\t$ at each time step $t$. Rather than directly dispatching the control signals $u\t$ to the actuators, GSK supplies $\delta\t$ to the actuators and relies on the actuators to compute $u\t$.

\begin{figure}
\centering
\includegraphics{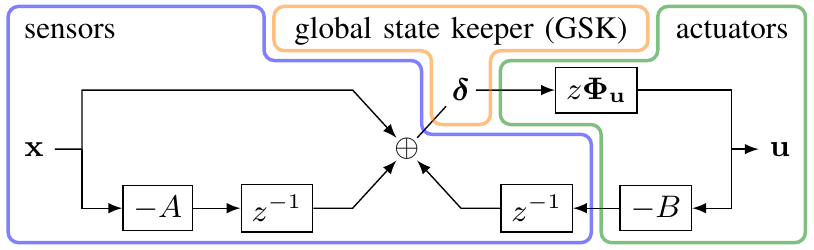}
\caption{The partitions of the block diagram for the global state architecture. \blue{The architecture of each part is depicted in \fig{global-state-architecture}.}}
\label{fig:global-state}
\vspace*{\baselineskip}
\subcaptionbox{The architecture of the global state keeper (GSK).}{
\includegraphics{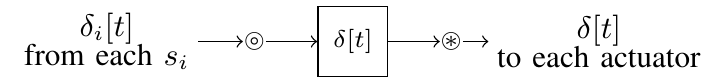}
}\\
\subcaptionbox{The architecture of each sensor $i$ ($s_i$).}{
\includegraphics{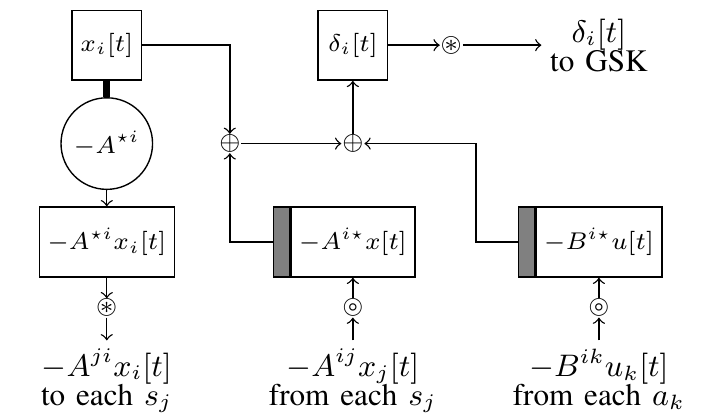}
}\\
\subcaptionbox{The architecture of each actuator $k$ ($a_k$).}{
\includegraphics{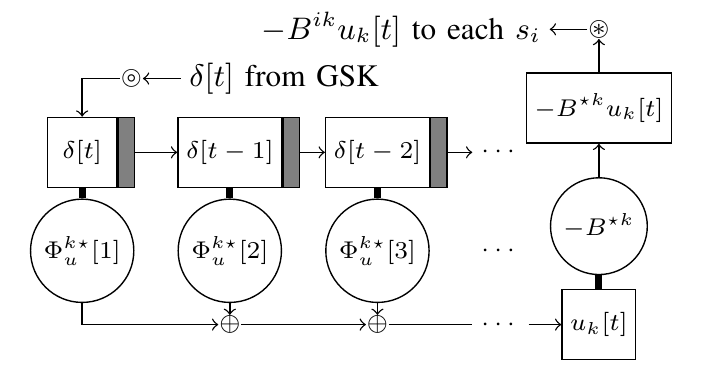}
}
\caption{The global state architecture.}
\label{fig:global-state-architecture}
\end{figure}

We illustrate the details of each node in \fig{global-state-architecture}. GSK collects $\delta_i\t$ from each sensor $s_i$. To compute $\delta_i\t$, each $s_i$ stores a column vector $-A^{\star i}$. Using the sensed state $x_i\t$, $s_i$ computes $-A^{ji}x_i\t$ and sends it to $s_j$. Meanwhile, $s_i$ collects $-A^{ij}x_j\t$ from each $x_j$ and $-B^{ik}u_k\t$ from each $a_k$. Together, $s_i$ can compute $\delta_i\t$ by
\begin{align*}
\delta_i\t =&\ x_i\t -A^{i\star}x\tm- B^{i\star} u\tm\\
=&\ x_i\t - \sum\limits_{j} A^{ij}x_j\tm - \sum\limits_{k} B^{ik}u_k\tm.
\end{align*}
The $\delta\t$ term is then forwarded to each actuator by GSK. The actuator $a_k$ can compute the control signal using the multiplier array similar to the structure in the centralized architecture. The difference is that each actuator only needs to store the rows of the spectral components ${\Phi_u^{k\star}}[\tau]$ of ${\Phibf_\ubf}$. After getting the control signal $u_k\t$, $a_k$ computes $-B^{ik}u_k\t$ for each sensor $s_i$.

One outcome of this communication pattern is that  $s_i$ only needs to receive $-A^{ij}x_j\t$ from $s_j$ if $A^{ij} \neq 0$. Similarly, only when $B^{ik} \neq 0$ does  $s_i$ need to receive information from $a_k$. This property tells us that the nodes only exchange information with their neighbors when a non-zero entry in $A$ and $B$ implies the adjacency of the corresponding nodes (as shown in \subfig{cyber-physical-comparison}{global-state}). Notice that this property holds for any feasible $\SFpair$, regardless of the constraint set $\Scal$.

The global state architecture mitigates the computation workload of the centralized architecture and hence improves the scalability. However, this architecture is subject to a single point of failure -- the GSK. Therefore, we explore some decentralized architectures in the following subsection.

%% file: distributed.tex
\subsection{Distributed Architectures}\label{sec:distributed}

To avoid the single point of failure, we can either reinforce the centralized unit by redundancy or deconstruct it into multiple sub-units, each responsible for a smaller region. Here we take the latter option to an extreme: We remove GSK from the partitions of the control diagram (\fig{global-state}) and distribute all control functions to the nodes in the network.

A naive way to remove GSK from the partitions is to send the state $\delta\t$ directly from the sensors to the actuators, as shown in \fig{distributed-naive}. More specifically, each sensor $s_i$ would send $\delta_i\t$ to all the actuators. Each actuator $a_k$ then assembles $\delta\t$ from the received $\delta_j\t$ for all $j$ and computes $u_k\t$ accordingly.

\begin{figure}
\centering
\includegraphics{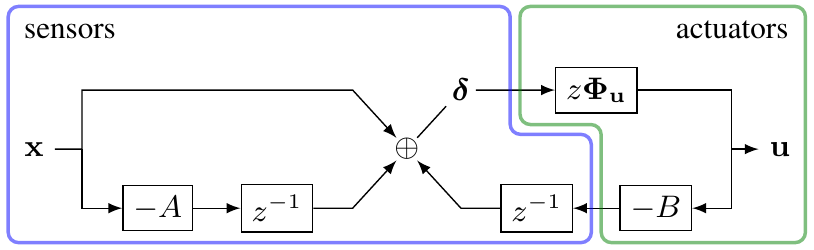}
\caption{A naive way to partition the block diagram in a distributed manner: \blue{Keep the architecture in \fig{global-state-architecture} while letting} the sensors directly send $\delta\t$ to each actuator. Those duplicated copies of $\delta\t$ waste memory resources.}
\label{fig:distributed-naive}
\end{figure}

This approach avoids the single point of failure problem. However, it stores duplicated copies of $\delta_i\t$ at each actuator, which wastes memory resources. To conserve  memory usage, we propose to send processed information to the actuators instead of the raw state $\delta\t$. We depict such a memory-conservative distributed scheme in \fig{distributed}.

\begin{figure}
\centering
\includegraphics{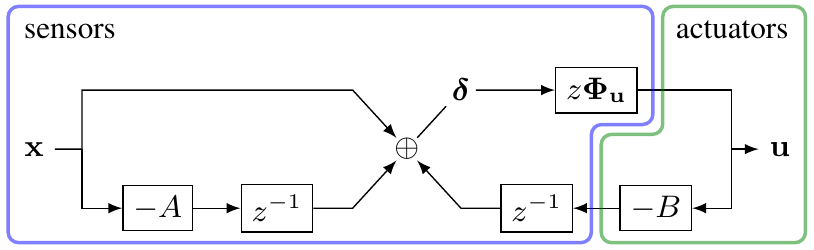}
\caption{The partition of the block diagram for the distributed architecture that conserves memory usage. \blue{The architecture of each part is depicted in \fig{distributed-architecture}.}}
\label{fig:distributed}
\vspace*{\baselineskip}
\subcaptionbox{The architecture of each sensor $i$ ($s_i$).}{
\includegraphics{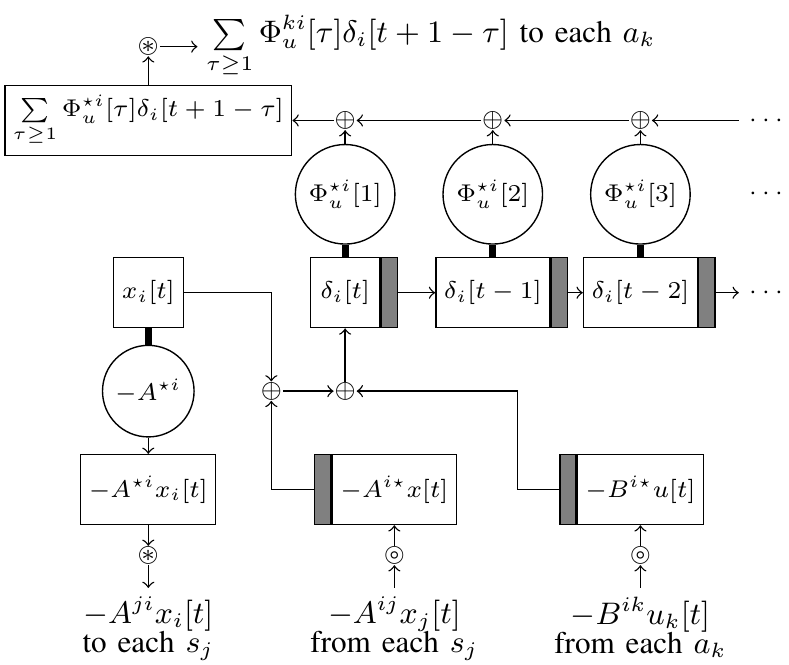}
}\\
\subcaptionbox{The architecture of each actuator $k$ ($a_k$).}{
\includegraphics{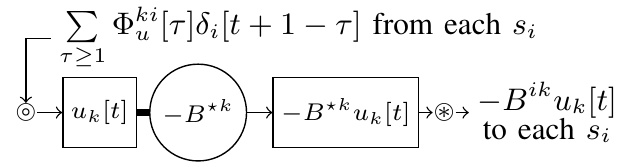}
}
\caption{The memory conservative distributed architecture.}
\label{fig:distributed-architecture}
\end{figure}

The only difference between \fig{distributed-naive} and \fig{distributed} is that we move the convolution $z{\Phibf_\ubf}$ from the actuator side to the sensor side. Implementation-wise, this change leads to the architectures in \fig{distributed-architecture}.

In \fig{distributed-architecture}, each sensor $s_i$ not only computes $\delta_i\t$, but $s_i$ also feeds $\delta_i\t$ into a multiplier array for convolution. Each multiplier in the array stores the $i^{\rm th}$ column of a spectral component of ${\Phibf_\ubf}$. The convolution result is then disseminated to each actuator.

At the actuator $a_k$, the control signal $u_k\t$ is given by the sum of the convolution results from each sensor:
\begin{align*}
u_k\t =&\ \sum\limits_{\tau \geq 1} {\Phi_u^{k\star}}[\tau]\delta[t+1-\tau]\\
=&\ \sum\limits_{i} \sum\limits_{\tau \geq 1} {\Phi_u^{ki}}[\tau]\delta_i[t+1-\tau].
\end{align*}

To confirm that \fig{distributed} is more memory efficient than \fig{distributed-naive}, we can count the number of scalar memory locations  in the corresponding architectures. Since both architectures store the matrices $A$, $B$, and $\Phi_u[\cdot]$ in a distributed manner, the total number of scalar memory locations needed for the multipliers is
\begin{align}
N_x^2 + N_x N_u + T N_x N_u.
\label{eqn:distributed-storage-requirement-multiplier}
\end{align}
We now consider the memory requirements for the buffers. Suppose ${\Phibf_\ubf}$ is FIR with horizon $T$ (or, there are $T$ multipliers in the $z{\Phibf_\ubf}$ convolution). In this case \fig{distributed-naive} has the same architecture as \fig{global-state-architecture} without the GSK.
Therefore, the total number of stored scalars (buffers) is as follows:
\begin{align}
\text{At each $s_i$:} &\ N_x + 4, \nonumber\\
\text{At each $a_k$:} &\ (T + 1) N_x + 1, \nonumber\\
\text{Total:} &\ N_x(N_x + 4) + N_u ((T + 1) N_x + 1)\nonumber\\
&\ = (T + 1) N_x N_u + N_x^2 + 4 N_x + N_u.
\label{eqn:distributed-naive-storage-requirement-buffer}
\end{align}
On the other hand, the memory conservative distributed architecture in \fig{distributed-architecture} uses the following numbers of scalars:
\begin{align}
\text{At each $s_i$:} &\ N_x + N_u + T + 3, \nonumber\\
\text{At each $a_k$:} &\ N_x + 1, \nonumber\\
\text{Total:} &\ N_x(N_x + N_u + T + 3) + N_u (N_x + 1)\nonumber\\
&\ = 2 N_x N_u + N_x^2 + (T + 3) N_x + N_u.
\label{eqn:distributed-storage-requirement-buffer}
\end{align}
In sum, the memory conservative distributed architecture requires fewer scalars in the system, and the difference is
$(T - 1) N_x (N_u - 1),$
which is non-negative since $T$, $N_x$, and $N_u \geq 1$.

Future work will look at specializing the above computation/memory costs to specific localization constraints. The above expressions serve as upper bounds that are tight for systems that difficult to localize in space (in the sense of $(d,T)$-localization defined in~\cite{WanMD18}). For systems localizable to smaller regions, the actuators only collect local information, and hence we don't need every sensor to report to every actuator.

%% file: comparison.tex
\section{Architecture Comparison}\label{sec:comparison}

\begin{table*}
\centering
\caption{Comparison Amongst the Proposed Architectures}
\def\acol{0.36\columnwidth}
\def\bcol{0.35\columnwidth}

\label{tab:comparison}
\renewcommand{\arraystretch}{1.25}
\begin{tabular}{
|>{\centering}m{\acol}
|>{\centering}m{\bcol}
|>{\centering}m{\bcol}
|>{\centering}m{\bcol}
|>{\centering}m{\bcol}
|}
\hline
\vspace*{-0.4\baselineskip}
\diagbox[width=0.41\columnwidth]{Property}{Architecture}
& \textbf{Centralized }(\fig{centralized})
& \textbf{Global State} (\fig{global-state})
& \textbf{Naive Distributed }(\fig{distributed-naive})
& \textbf{Memory Conservative Distributed }(\fig{distributed})
\tabularnewline
\hline
\hline
\textbf{Single Point of Failure}
& yes
& yes
& no
& no
\tabularnewline
\hline
\textbf{Overall Memory Usage}
& lowest\\
(see \eqn{centralized-storage-requirement})
& highest\\
(equal to \eqn{distributed-storage-requirement-multiplier}+\eqn{distributed-naive-storage-requirement-buffer}+$N_x$)
& second highest\\
(equal to \eqn{distributed-storage-requirement-multiplier}+\eqn{distributed-naive-storage-requirement-buffer})
& second lowest\\
(equal to \eqn{distributed-storage-requirement-multiplier}+\eqn{distributed-storage-requirement-buffer})
\tabularnewline
\hline
\textbf{Single Node\\Memory Usage}
& high
& low (actuator needs more memory)
& low (actuator needs more memory)
& low (sensor needs more memory)
\tabularnewline
\hline
\textbf{Single Node\\Computation Loading}
& high (see \eqn{centralized-computation-requirement})
& low (actuator performs convolution)
& low (actuator performs convolution)
& low (sensor performs convolution)
\tabularnewline
\hline
\textbf{Single Node\\Communication Loading}
& sensor/actuator: low\\
controller: high
& sensor/actuator: medium\\
GSK: high
& sensor/actuator:\\high, but localizable
& sensor/actuator:\\high, but localizable
\tabularnewline
\hline
\end{tabular}
\end{table*}

\begin{figure*}
\centering
\subcaptionbox{Centralized Architecture}{
\includegraphics{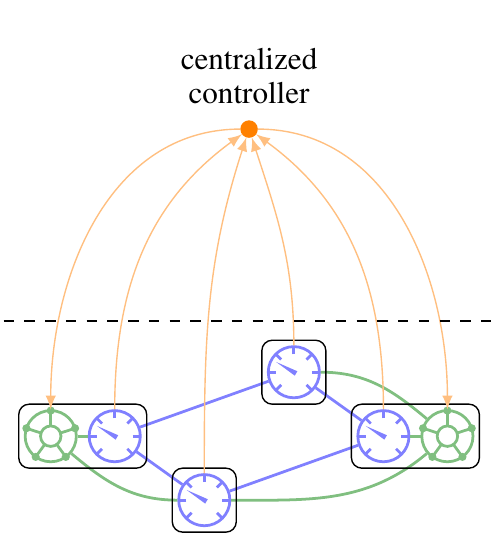}
}\quad\quad
\subcaptionbox{Global State Architecture\label{subfig:cyber-physical-comparison-global-state}}{
\includegraphics{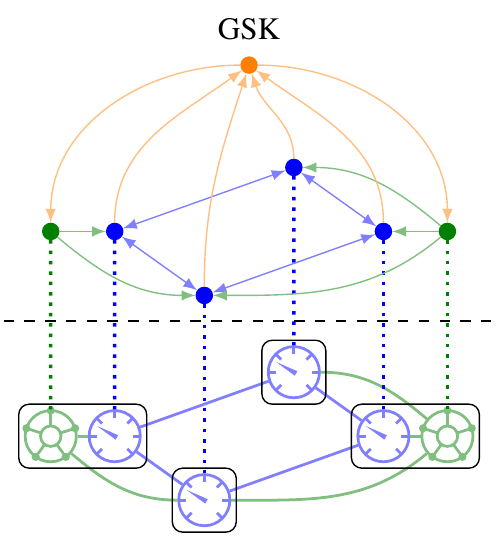}
}\quad\quad
\subcaptionbox{Distributed Architectures}{
\includegraphics{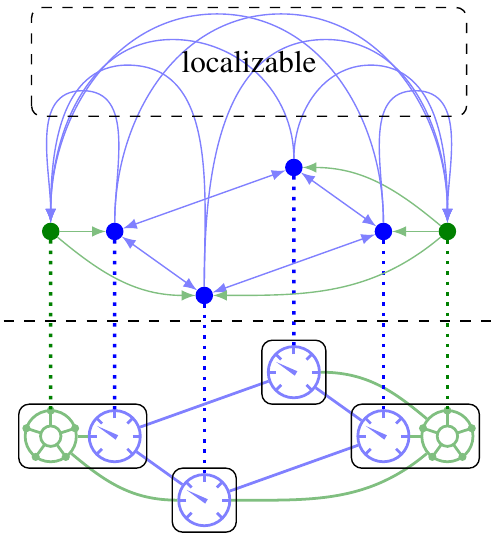}
}
\caption{The cyber-physical structures of the proposed architectures. The horizontal dashed line separates the cyber (top) and the physical (bottom) parts of the system. Each solid node represents a computation unit, and the arrow links are the communication channels. The dotted lines associate the computation units with their locations. The cyber structure of the centralized architecture is ignorant about the underlying physical system, while the other architectures manifest some correlations between the cyber and the physical structures. The distributed architectures (naive and memory conservative) replace GSK by direct communications, which can be trimmed by imposing appropriate localization constraints on ${\Phibf_\ubf}$ during the synthesis phase.}
\label{fig:cyber-physical-comparison}
\end{figure*}

As mentioned at the beginning of \sec{architecture}, different architectures implementing the same controller model allow the engineer to consider different trade-offs. Here we compare the proposed architectures and discuss their differences. Our findings are summarized in \tab{comparison}.

In terms of robustness, the centralized and the global state architectures suffer from a  single point of failure, i.e., the loss of the centralized controller or the GSK paralyzes the whole system. This also makes the system vulnerable from a cyber-security perspective. On the contrary, the distributed architectures can still function with some nodes knocked out of the network.\footnote{The system operator might need to update  ${\Phibf_\ubf}$ to maintain  performance --  such a re-design is well within the scope of the SLS framework.}

For information storage, the centralized controller uses the fewest buffers, and the global state architecture stores the most variables (because its GSK has to relay $\delta\t$). We remark that although the centralized scheme achieves the minimum storage usage at the system level, the single node memory requirement is high for the centralized controller. Conversely, the other architectures store information in a distributed manner, and a small memory is sufficient for each node.

We evaluate the computational load at each node by counting the number of performed multiplication operations. The centralized architecture aggregates all the computation at the centralized controller, while the other architectures perform distributed computing. For distributed settings, the computation overhead is slightly different at each node. The global state and naive distributed scenarios let the actuators compute the convolution. Instead, the memory conservative distributed architecture puts the multiplier arrays at the sensors.

Finally, we discuss the communication loading of the architectures. In \fig{cyber-physical-comparison}, we sketch the resulting cyber-physical structures of each scheme. The centralized architecture ignores the underlying system interconnection and installs a centralized controller to collect information and dispatch control actions. Under this framework, the sensors and the actuators only need to recognize the centralized controller, but the centralized controller must keep track of all the nodes in the system, which limits the scalability of the scheme.

The global state architecture also introduces an additional node into the system, the GSK, with which all  nodes should be contact with. Meanwhile, the sensors and actuators also communicate with each other according to the matrices $A$ and $B$. In other words, if two nodes are not directly interacting with each other in the system dynamics, they don't need to establish a direct connection.

Similarly, in the distributed architectures, the sensors and the actuators maintain connections according to $A$ and $B$. Additionally, direct communications, which are governed by the structure of ${\Phibf_\ubf}$, are added to replace the role of GSK. Although it would be slightly more complicated than having a GSK as the relay, we can localize ${\Phibf_\ubf}$ at the synthesis phase to have a sparse communication pattern.

Besides the centralized architecture, the physical structure has a direct influence on the cyber structure in all other architectures. As such, we believe further research on the deployment architectures of SLS would lead to better cyber-physical control systems.

%% file: future-directions.tex
\section{Conclusion and Future Directions}\label{sec:future}
A new internally stabilizing state-feedback controller was derived for systems that are open-loop stable. The  controller was shown to have a block diagram realization that is in some ways simpler than the ``standard'' SLS state-feedback controller.

We considered various architectures to deploy this controller to a real CPS. We illustrated and compared the memory and computation trade-offs among four different deployment architectures: centralized, global state, naive distributed, and memory conservative distributed architectures.

Future work involves removing the open-loop stability requirement and considering the output-feedback setting. Also, as pointed out in \sec{distributed}, there are still many decentralized architecture options left to explore. For example, in addition to the distributed schemes, one can install multiple local controllers to form clustered architectures and employ localization constraints to limit information exchange among clusters.